\documentclass[12pt]{article}
\usepackage[T1]{fontenc}
\usepackage{kpfonts}
\usepackage[dvipsnames]{xcolor}
\usepackage{pdfpages}
\usepackage{sgame}
\usepackage{accents}
\usepackage{tikz-cd}
\usepackage{float}
\usepackage[round]{natbib}
\usepackage{multirow}
\usepackage{dutchcal}
\usepackage{pdfpages}
\usepackage{bbm}
\usepackage{sgame}
\usepackage{mathtools}
\usepackage{amsthm}
\usepackage{enumitem}
\usepackage[margin=1in]{geometry}
\usepackage{caption}
\usepackage{subcaption}
\usepackage{epigraph}
\usepackage{listings}
\usetikzlibrary{calc}
\usetikzlibrary{shapes,arrows}
\usepackage{tcolorbox}

\newcommand{\ubar}[1]{\underaccent{\bar}{#1}}
\DeclareMathOperator\supp{supp}

\newtheorem{theorem}{Theorem}[section]
\newtheorem{proposition}[theorem]{Proposition}
\newtheorem{lemma}[theorem]{Lemma}
\newtheorem{corollary}[theorem]{Corollary}
\newtheorem{claim}[theorem]{Claim}

\theoremstyle{definition}

\newtheorem{definition}[theorem]{Definition}

\usepackage{hyperref}
\hypersetup{
    colorlinks=true,
    linkcolor=Aquamarine,
    filecolor=Aquamarine,      
    urlcolor=Aquamarine,
    citecolor=Aquamarine,
}

\definecolor{backcolour}{rgb}{0.63, 0.79, 0.95}

\lstdefinestyle{mystyle}{
  backgroundcolor=\color{backcolour},
  basicstyle=\ttfamily\footnotesize,
  breakatwhitespace=false,         
  breaklines=true,                 
  captionpos=b,                    
  keepspaces=true,                 
  numbers=left,                    
  numbersep=5pt,                  
  showspaces=false,                
  showstringspaces=false,
  showtabs=false,                  
  tabsize=2
}

\providecommand{\keywords}[1]{\textbf{\textit{Keywords:}} #1}
\providecommand{\jel}[1]{\textbf{\textit{JEL Classifications:}} #1}

\begin{document}
\title{Attraction Via Prices and Information}
\author{Pak Hung Au\thanks{Hong Kong University of Science and Technology. Email: \href{mailto:aupakhung@ust.hk}{aupakhung@ust.hk}.} \and Mark Whitmeyer\thanks{Arizona State University. Email: \href{mailto:mark.whitmeyer@gmail.com}{mark.whitmeyer@gmail.com}. We began this project while the second author was at the Hausdorff Center for Mathematics at the University of Bonn, supported under the DFG Project 390685813. This paper subsumes \cite{dia}. We thank Ralph Boleslavsky, Francesc Dilm\'{e}, Silvana Krasteva, Stephan Lauermann, Benny Moldovanu, Franz Ostrizek, Wing Suen, Joseph Whitmeyer, Kun Zhang, and various seminar and conference audiences for their useful feedback.}}

\date{\today}

\maketitle

\begin{abstract}
We study the ramifications of increased commitment power for information provision in an oligopolistic market with search frictions. Although prices are posted and, therefore, guide search, if firms cannot commit to information provision policies, there is no active search at equilibrium so consumers visit (and purchase from) at most one firm. If firms can guide search by \textit{both} their prices and information policies, there exists a unique symmetric equilibrium exhibiting price dispersion and active search. Nevertheless, when the market is thin, consumers prefer the former case, which features intense price competition. Firms always prefer the latter.
\end{abstract}
\\
\jel{C72; C73; D15; D63}

\newpage

\section{Introduction}

There are many markets in which prices guide search. Before she leaves her home to go shopping for a car or before she orders a dress online (which she will return if ill-fitting), a consumer visits a price-comparison website--or compares print advertisements distributed by various sellers--to see who is offering the best deal. Shopping is costly (both in terms of time and money) and by gathering information about prices beforehand, the consumer can search efficiently, without wasted effort.

In addition to being able to set prices, firms often also have significant control over how much match-value information a consumer’s visit will yield. They can choose whether to offer test drives, whether to allow consumers to try on clothes, the generosity and flexibility of return policies, and the length (and functionality restrictions) of free trials for software. Moreover, in many markets, this aspect of search is also advertised by firms. Just as price-comparison websites allow consumers to compare firms’ price offers, there exist similar websites that inform consumers about firms’ information provision. Therefore, consumers know not only the collection of price offers, but also the various information policies offered by firms, both of which are valuable in guiding the consumers’ search for a satisfactory product.


\bigskip

\noindent \textbf{Single-instrument competition (what has been studied so far):} By now, it is understood that allowing firms to post prices when information is exogenous, by and large, benefits consumers. Firms value being visited early in a consumer’s journey--for a position late in a consumer’s search order may mean that the firm is never even visited--and a low price helps a firm secure an early visit. With exogenous information and posted prices, search frictions intensify competition as firms compete in a quasi-Bertrand manner (\cite*{choi2}). Furthermore, as established in \cite*{avp}, a partial analog holds for this result when the price is exogenous and firms compete by posting information. As firms that provide lots of information are visited earlier by consumers, search frictions can encourage information provision, to the benefit of consumers.

In the price-competition literature and in \cite{avp}, there is a clear welfare comparison between the scenarios when prices are posted versus not and when information is posted versus not. The posted-price environment benefits the consumers by directing their search to firms charging lower prices, which in turn forces firms to compete for their visits by cutting prices. Following a similar intuition, in the information-competition setting, the posted-information environment benefits the consumers. If the information provision policies are not posted, there is an ``informational Diamond paradox:’’ firms provide no useful information and the market breaks down.


\bigskip

\noindent \textbf{Multi-instrument competition (this paper's setting):} Both the price-competition literature and \cite*{avp} fix one instrument (prices or information) and concentrate on firms competing by choosing the other. In plenty of settings, this is reasonable: certain platforms limit price competition, for instance; and in other markets, firms can control prices but information about their products (at least to a first-order approximation) is exogenous. Nevertheless, there are other settings in which firms can likely control both. Our goal in this paper is to understand competition in an oligopolistic environment with search frictions in which firms compete both by setting prices and flexibly choosing the information a consumer’s visit will reveal. 

We stipulate that prices are posted--and, therefore, direct search--but look at both the setting in which information is also posted as well as the scenario in which only prices are. Our main result is that when the market is thin, consumer welfare is higher when firms cannot post (and advertise) their information provision. Otherwise, when the market is thick, consumers prefer that both instruments be advertised. Firms always prefer that information is posted.

How can this be? When competition is limited to a single instrument, the relationship is clear: when firms compete to attract consumers, be it through prices or information, consumers benefit. The key to understanding our result here is that the ability of firms to post information not only affects the information they provide but also the prices they set. We show that when firms can only post prices, the ``informational Diamond paradox’’ persists: firms provide no useful information in any symmetric equilibrium and there is no search. However, firms still post prices, and the fact that there is no search actually \textit{intensifies} the price competition to the maximal level. Standard Bertrand competition ensues, where the market price is driven down all the way to the firms’ marginal cost. Thus, in such a market, the consumers enjoy low prices but no information is provided.

When both instruments are attraction devices, we characterize the unique symmetric equilibrium and show that it involves active search and price dispersion. Moreover, firms provide full information, which is the polar opposite of what they do in the hidden information case. This benefits the consumer by allowing her to obtain a better match, and sooner in her search. If the market is sufficiently thick, price competition is already intense enough and the consumers prefer the environment with posted information to that without.



\subsection{Related Literature}

Although the majority of the consumer search literature assumes that prices do not guide search and that consumers search randomly, a number of recent papers explore how things change when consumers' searches are directed and guided by firms' price posting. These papers include \cite*{az}, \cite*{ding}, \cite*{haan}, and \cite*{choi2}. Of particular note is \cite*{ding}, who have binary match values and consumer heterogeneity (\'{a} la \cite*{stahl}). Although we do not allow for consumer heterogeneity, we endogenize the match value distribution and show that the one proposed in \cite*{ding} is the equilibrium outcome even if firms have total freedom over consumers' (expected) match value distributions. 

The literature on information design in search markets is also related. The closest paper in this literature is our previous paper, \cite*{avp}, in which we study competitive information provision in a search market but abstract away the price instrument. \cite*{he2023competitive} show that the ``informational Diamond paradox'' identified in \cite*{avp} persists even when signals are public, provided consumers cannot direct their search. Also related is \cite*{hu}, who explore consumer-optimal information structures in a large-market version of \cite*{asher}'s undirected search setup. \cite*{pease} and \cite{lyu2023information} look at \textit{pre-visit} information about a seller's good.

\cite{krasteva} also study information provision and pricing in search markets. Crucially, their environment is one \textit{without} price comparison websites: following the \cite{asher} paradigm, search is undirected and consumers observe prices upon visiting firms. They also incorporate the \cite{stahl} setup of having a mix of zero search-cost ``shoppers'' and costly searchers, which is needed to stymy the Diamond paradox. Notably, they obtain full-disclosure only when search costs are very low or high. \cite*{gar} and \cite*{hkb} also explore an environment in which firms compete through price setting and information provision, but in a market absent search frictions. That is, in their papers consumers freely observe every firm's price and signal realization.

\cite{sato2023information} and \cite{sato2023persuasion} both study Bayesian persuasion by a planner in a directed search environment. The former characterizes the set of possible choice probabilities for each sender in the market. The planner in the latter maximizes a sum of per-visit fees. \cite*{ke2022information}'s planner is a platform that extracts surplus from sellers who vie for prominent positions. \cite*{mekonnen2023persuaded} also explore a planner's problem, but in a large-market undirected search setting. There, the planner endeavors to string a searcher along: it charges a per-period fee, seeking simultaneously to prolong search and extract rents.

Finally, there has been a recent explosion of papers studying information design by a planner in oligopolistic markets. One such paper is \cite*{elliott2021market}, who study the possible market outcomes an information designer can produce. \cite*{bergemann2015limits}, in turn, conduct such an exercise in a market with a single monopolist. In \cite*{elliott2022matching}, the information designer can also control the \textit{matchings} of firms with consumers; and in \cite{limismatch} \textit{only} the matchings can be controlled (so there is no information design). An important distinction between these works and ours is that we have multiple strategic (and competing) information designers, who can only control information about their idiosyncratic match values.
\section{Model}

There are \(n\) identical firms supplying horizontally differentiated products, and one representative consumer with unit demand. The match values of the firm's products with the consumer are independently and identically distributed random variables, denoted by \(x_i\) for firm \(i\), taking values either \(0\) or \(1\), with \(\mu \equiv \mathbb{P}\left(1\right) \in \left(0,1\right)\). For simplicity, the marginal cost of production is \(0\), and the consumer's outside option is also \(0\). The consumer conducts a sequential search for a satisfactory product produced by (at most) one of the \(n\) firms. 

Each firm \(i\) simultaneously sets a price \(p_i \geq 0\) and chooses a signal. It is standard to model the chosen signal as a 
distribution over posterior (expected) values \(F_i\) that must be Bayes-plausible, i.e., \(\mathbb{E}_{F_i}\left[x\right] = \mu\) and \(F_i\) is supported on a subset of \(\left[0,1\right]\). Denote the set of Bayes-plausible distributions over posterior match values by  \(\mathcal{F}\). The distribution over posteriors chosen by firm \(i\), \(F_i\), determines the information that the consumer can acquire upon visiting the firm after incurring a search cost of \(c \in \left[0,\mu\right)\). Specifically, after paying the search cost, the consumer observes the signal realization of the firm and forms a posterior belief about her match value with its product.

The consumer's sequential search is directed by available market information. We focus on two cases. In the first--which we term \textbf{Hidden Information}--firms can only publicly post their prices, but not their information provision policies, at the outset of the consumer's search. We make the mild stipulation of ``no signaling what you don't know'' about the consumer's beliefs: upon seeing a deviation by a firm, the consumer's belief about the information provided by other firms is unchanged. With hidden information, the consumer's sequential search is guided by the collection of posted prices and the consumer's conjectures about firms' information provision. In the second case--which we term \textbf{Posted Information}--not only the firms' prices but also their signals (though not the signal realizations) are posted publicly. Therefore, both the information policies and prices influence the consumer's sequential search. 

The chosen firm \(i\) (if any) receives a payoff of \(p_i\), and the others, \(0\). The consumer receives (expected) payoff \(x_i - p_i - k c\) if she purchases from some firm \(i\) or \(- k c\) if she takes her outside option, where \(0 \leq k \leq n\) is the number of firms she visits before her stopping decision. We also impose that this is search with costless recall: the consumer can return at no cost to any firm whom she has already visited, in order to purchase its product. Moreover, the consumer may not purchase from any theretofore unvisited firm.\footnote{This assumption is standard in the literature and allows us to use the apparatus of \cite*{wei} to characterize the search problem in a parsimonious way.}

The timing of the game is as follows:
\begin{enumerate}[noitemsep,topsep=0pt]
    \item Firms simultaneously set prices and choose distributions of the posterior match values.
    \item The consumer observes the collection of prices (and match value distributions, in the posted information case).
    \item The consumer embarks on her search, deciding whom to visit and when to stop.
\end{enumerate}

Firms can play mixed strategies over both prices and signals. If they do, the randomization is resolved before the consumer embarks on her search. We adopt the solution concept of symmetric weak perfect Bayesian Nash equilibrium with the aforementioned ``no signaling what you don't know'' specification. 

\subsection{Simplifying the Consumer's Search Problem}\label{simplify section}
Our goal is to determine the equilibrium pricing and information provision by the firms. To this end, we work backward and first describe the consumer's search behavior, which pins down each firm's demand. The seminal result of \cite*{wei} identifies the consumer's optimal behavior. For each distribution \(F_{i}\) and price \(p_i\) chosen by firm \(i\), we define the corresponding reservation value, \(U\left(F_{i}, p_i\right)\), implicitly as the solution to the following equation (in \(U\)): 
\[c=\int_{U+p_i}^{1}\left(x-p_i-U\right) dF_{i}\left(x\right) \text{ .}\label{reservation value} \tag{\(1\)}\]

The set of feasible reservation values is \[\left\{ U\left(F_{i}, p_i\right) \colon F_{i}\in \mathcal{F}, p_i \geq 0\right\} =\left[ \ubar{U}-p_i,\bar{U}-p_i\right] \text{,}\] where \(\ubar{U}-p_i \equiv \mu -c-p_i\) and \(\bar{U}-p_i \equiv \frac{\mu-c}{\mu}-p_i\). The lower bound is induced by any Bayes-plausible distribution whose support is entirely (weakly) above \(\ubar{U}\), one of which is the degenerate distribution at \(\mu\) (which corresponds to a completely uninformative signal). The upper bound is uniquely induced by the Bayes-plausible distribution supported on \(\left\{0,1\right\}\) (which corresponds to a fully revealing signal). It is not difficult to see that any intermediate reservation value can be achieved by some Bayes-plausible distribution.\footnote{For instance, by some convex combination of the two distributions above.} The reservation value rewards informativeness: for any pair of feasible distributions \(F_{i}\) and \(G_{i}\), if \(F_{i}\) is a mean-preserving spread of \(G_{i}\), then \(U\left( G_{i},p_i\right) \leq U\left( F_{i},p_i\right)\). 

The optimal strategy of the consumer (\hypertarget{pandora}{\textbf{Pandora's rule}}) is as follows.
\begin{itemize}[noitemsep,topsep=0pt]
\item Selection rule: If a firm is to be visited and examined, it should be the unvisited firm with the highest reservation value.
\item Stopping rule: Search should be stopped whenever the maximum reservation value of the unvisited firms is lower than the maximum sampled
reward or the outside option.
\end{itemize}
It turns out that the consumer's search problem can be simplified even further. We define the effective value of firm \(i\) to be the random variable 
\[W_{i}\equiv \min \left\{ x_i-p_{i},U\left(F_{i},p_i\right) \right\} \text{ ,}\] where \(x_i - p_{i}\in \left[ 0,1\right]\) is the realized posterior match value minus the price charged by firm \(i\) and \(F_{i}\in \mathcal{F}\) is the distribution over posteriors offered by firm \(i\). 

In a remarkable result, \cite*{choi2} and \cite*{am} show that the consumer purchases from firm \(i\) if and only if it has the highest effective value realization, i.e., \(W_{i}=\max_{j=1,2,...,n}W_{j}\), and that its effective value exceeds the consumer's outside option. Thus, we can reformulate the firms' competition to a simpler static problem: firms compete by setting prices and choosing distributions over effective values. 

Firm \(i\)'s effective value, \(W_{i}\), has distribution
\[H\left(w; F_{i}, p_i\right) \equiv \begin{cases}
F_{i}\left(w+p_i\right) \quad &\text{if} \quad -p_i \leq w < U\left(F_{i},p_i\right) \\ 
1 \quad &\text{if} \quad w \geq U\left( F_{i},p_i\right)
\end{cases}\text{ .}  \label{eff-value dist.}\tag{\(2\)}
\]
A firm's effective-value distribution must be supported on a subset of \(\left[-p_i, \bar{U}-p_i\right]\) and have mean \(\mu-c-p_i\). As noted by \cite*{avp}, there are further conditions this distribution must satisfy (see Lemmas 3.4 and 3.5 of \cite*{avp}).

\section{Hidden Information}\label{hiddeninfosection}

We begin with the case in which prices are posted but information policies are hidden. Given the collection of posted prices, the consumer forms a conjecture about the distribution over posteriors \(\tilde{F}_i\) at each firm \(i\). Naturally, the conjecture must be correct in equilibrium. Based on these posted prices and conjectures, the consumer then uses \hyperlink{pandora}{Pandora's} rule to determine her optimal search protocol. We specify ``no signaling what you don't know:'' any deviation by a firm leaves the consumer's beliefs about the other firms' strategies unchanged. Note that it is only here, when information is hidden, that this specification has bite--when everything is posted (in \(\S\)\ref{postedinfosection}), nothing is hidden so there are no conjectures to be made.

The following definition is useful in stating the main result of this section.
\begin{definition}
    A distribution over posteriors provides useless information if and only if its support is contained in \(\left[\ubar{U},1\right]\). Otherwise, we say the distribution provides useful information.
\end{definition}
A distribution with support contained in \(\left[\ubar{U},1\right]\) is useless information because its realization never has any effect on the consumer's optimal search strategy. If a firm pricing at \(p_i\) is expected to provide useless information, its conjectured reservation value is \(\ubar{U}-p_i\). If this value is negative, the firm is never visited. If this value is non-negative, conditional on a visit, the consumer must find it sequentially optimal to buy from this firm immediately regardless of the posterior realization. This is because under \hyperlink{pandora}{Pandora's} rule, the consumer visits firm \(i\) if and only if she has no options at hand that are better than firm \(i\)'s reservation value, and she sees no other unvisited firms with a higher reservation value. 

In contrast, if a distribution provides useful information, its realization can potentially affect the consumer's decision. If a firm pricing at \(p_i\) is expected to offer a signal \(\tilde{F}_i\) that provides useful information, its conjectured reservation value \(U\left(\tilde{F}_i, p_i\right)\) is strictly above \(\ubar{U}-p_i\). Moreover, by the reservation-value equation (\ref{reservation value}), \(\tilde{F}_i\) must assign positive probability to posteriors strictly above \(U\left(\tilde{F}_i, p_i\right)+p_i\), as well as to posteriors strictly below \(\ubar{U}+p_i\). In the former case, the consumer will stop her search immediately, whereas in the latter case, she may continue her search journey.     


The main result of this subsection is that when prices are posted but information is hidden, all firms post a zero price and provide useless information in equilibrium. The intuition is that firms' information provision policies do not guide search--instead it is the consumer's conjectures about the information provision policies that do--a firm can provide less information without hurting its position in the consumer's search order. In fact, firms have an incentive to assign as much weight as possible to the posterior just above the consumer's stopping threshold, resulting in useless information to the consumer--an ``\textit{informational Diamond paradox}.'' But if all firms provide no information, there is no product differentiation. As prices guide search, competition reduces to standard Bertrand competition; and, hence, the marginal cost pricing outcome.

The following lemma is key to the main result stated above.
\begin{lemma}\label{onlyonelemma}
    When information is hidden, in any symmetric equilibrium, firms provide useless information and the consumer visits only one firm.
\end{lemma}

The logic of the lemma is explained below. It is immediate that the consumer makes at least one visit in equilibrium--otherwise, a firm can profitably deviate to price below \(\ubar{U}\) to secure a visit and sale. Consider two possibilities concerning the equilibrium expected profit--whether it is positive or zero. In the case of positive profit, positive prices must be charged with probability one. In particular, suppose price \(p_i>0\) and distribution \(F_i\) is on the support of the firms' (possibly mixed) equilibrium strategy. Given \hyperlink{pandora}{Pandora's} rule, whenever the firm is able to realize a posterior no less than  \(U\left(F_{i}, p_i\right)+p_i\), it can convert all visits to sales. The firm's payoff function (per consumer visit) in the realized posterior \(x_i\) is, thus, equal to \(p_i\) (the highest possible level) for all \(x_i \geq U\left(F_{i}, p_i\right)+p_i\). 

Therefore, if the prior match value \(\mu<U\left(F_{i}, p_i\right)+p_i\), the firm's optimal signal is supported on \(\left[0, U\left(F_{i}, p_i\right)+p_i\right]\). The reservation-value equation (\ref{reservation value}) implies that the actual reservation value of firm \(i\) is strictly below \(U\left(F_{i}, p_i\right)\), a contradiction. If \(\mu \geq U\left(F_{i}, p_i\right)+p_i\), the optimal signal must ensure the consumer stops and buys with probability one (for instance, those supported on \(\left[U\left(F_{i}, p_i\right)+p_i,1\right]\) meaning that it is useless information). In equilibrium, the consumers' belief about the firm's hidden information policy must be correct. Consequently, she anticipates no useful information from any firm, and never visits more than one firm.

Consider next the case of zero equilibrium profit, in which case every firm posts a zero price. Suppose for the sake of contradiction that with positive probability, some signal with a reservation value strictly above \(\ubar{U}\) is offered in equilibrium.\footnote{This allows for the possibility that the firm plays a mixed strategy in distributions over posteriors.} As posteriors strictly below \(\ubar{U}\) are realized with a positive probability, there exists a profitable deviation--a (small) positive price may still yield strictly positive sales. To see this, suppose the equilibrium has each firm realizing posteriors strictly below \(\ubar{U}-\varepsilon\) with some positive probability \(\delta\). Then if an individual firm deviates to price at \(\varepsilon/2\), its expected profit is strictly positive. This is because even if the consumer entertains the worst possible belief about the information offered by the deviating firm, the conjectured reservation value of the firm is still strictly above \(\ubar{U}-\varepsilon\), so the consumer is willing to visit it after observing posteriors below \(\ubar{U}-\varepsilon\) at all the other firms, which is an event with probability at least \(\delta^{n-1}\) by hypothesis. Moreover, the deviating firm can guarantee a positive expected profit conditional on the consumer's visit by, say, offering no information. Therefore, the only possible equilibrium in this case has all firms offering useless information. For the same reason as the case above, the consumer makes one and only one visit.

As no firm provides useful information, competition reduces to plain vanilla undifferentiated Bertrand competition, in which the unique equilibrium outcome is to price at marginal cost, which in this case is \(0\). Consequently, a symmetric equilibrium necessarily has all firms pricing at \(0\) and providing useless information. It can be readily verified that this strategy indeed constitutes a symmetric equilibrium, which yields 
\begin{proposition}\label{diamond}
When information is hidden, there is a(n essentially)\footnote{The distribution over effective values in any symmetric equilibrium is degenerate and the equilibrium price equals firms’ marginal cost, which is \(0\). Any randomization over actual values supported entirely above \(\mu - c\) yields this effective-value distribution and all are behaviorally equivalent.} unique symmetric equilibrium. In it, firms provide useless information and charge price \(p = 0\). The consumer stops and buys from the first visited firm. 
\end{proposition}\begin{proof} 
    It is easy to sustain the strategy described as an equilibrium. Upon seeing a strictly positive price by a firm (which is a deviation), the consumer conjectures that it provides no information. The deviator is never visited, so its profit remains \(0\). There are obviously no profitable deviations to other distributions given the market price of \(0\).

    Uniqueness is an immediate consequence of Lemma \ref{onlyonelemma} As the equilibrium necessarily has all firms supplying useless information and the consumer visiting only one firm, only the firm(s) with the lowest posted price will be visited. Every nondegenerate price distribution is suboptimal as prices close to the supremum of the support yield profit arbitrarily close to \(0\). The standard logic of Bertrand competition of undifferentiated products then implies the only equilibrium price is \(0\).
    \end{proof}

\subsection{No Posting = No Commitment}

This paper inhabits the Bayesian-persuasion paradigm (\cite{kam}), in which information designers have the power to commit to information policies. What if the firms did not have this commitment power, but were instead privately informed and could only transfer information via cheap talk? It turns out that when information is hidden, it makes no difference whether firms can commit to information policies at all. That is, with search frictions, the inability to post information policies is equivalent to the inability to commit to information policies.

More formally, suppose that each firm is \textit{ex ante} informed about the match value of its product.\footnote{When all communication is via cheap talk, it makes no difference how informed firms are about their match values.} As before, firms simultaneously set prices, which guide search. However, now each firm of each type chooses a distribution over messages, the realization of which is revealed to the consumer upon visiting that firm.
\begin{proposition}
When communication is via cheap-talk, there is a(n essentially) unique symmetric equilibrium. In it, firms provide useless information and charge price \(p = 0\). 
\end{proposition}\begin{proof}
    Consider first the case of positive equilibrium profit. Given a positive price is charged, it is immediate the only continuation equilibrium in the communication subgame (conditional on visit) is babbling -- the sender's incentive compatibility implies that every type pools to the message with the highest probability of sales.

    The case of zero equilibrium profit is similar to that analyzed in Lemma \ref{onlyonelemma}. If useful information were communicated with positive probability, a profitable deviation to a positive price necessarily exists. Note here the communication may not be complete babbling but it must be so uninformative that the consumer never finds it worthwhile to search beyond the first visited firm.
    
    
    Via the same logic explained in Proposition \ref{diamond}, it is easy to sustain zero price and useless information as an equilibrium.\end{proof}

\section{Posted Information}\label{postedinfosection}

This section considers the case in which both the prices and the information policies are posted, so the consumer's search is directed by both. As
explained in \(\S\)\ref{simplify section}, the consumer eventually purchases from the consumer
with the highest effective value, so the competition between firms can be
viewed as the competition over effective values. Specifically, each firm's
objective is to maximize the probability that its effective value
realization is the highest among all the \(n\) firms (as well as above the
consumer's outside option). In this competition, a typical strategy of firm \(i\) takes the form \[\sigma _{i}\equiv \left( \Phi _{i},\left\{ H_{i}\left(
\cdot ;p\right) \colon p\in \supp\left( \Phi _{i}\right) \right\} \right)\text{ ,}\] where \(\Phi_{i}\in \Delta \left[ 0,1\right]\) is firm \(i\)'s price distribution,
and \(H_{i}\left( w;p\right) \) is the effective-value distribution induced by
the firm's signal strategy (possibly mixed) conditional on price \(p\in
\supp\left( \Phi _{i}\right)\). The effective value distribution induced by
the firm's strategy \(\sigma_i\) is, thus, \(\int_{0}^{1}H_{i}\left(
w;p\right) d\Phi _{i}\left( p\right)\).

In contrast to the hidden information case, the equilibrium profit must be positive when information is also posted. A firm can guarantee itself a strictly positive expected profit by, e.g.,
pricing at \(\frac{\bar{U}-\ubar{U}}{2} > 0\) and offering full
information. This produces reservation value \(\frac{\bar{U}-\ubar{U}}{2}\), which allows the firm to secure a sale with positive
probability and, hence, a positive profit (recall the mean of the effective value distribution is at most \(\ubar{U}\)).

Another major difference from the hidden information case is that it is
almost always a bad idea to offer no information when posting information is
an option. Intuitively, giving up the power of using one's information policy to attract the consumer is very costly.

\begin{lemma}
\label{no info dominated}When firms post both information policies and
prices, a strategy consisting of no information is weakly dominated. The
domination is strict unless its price is degenerate at \(\ubar{U}\).
\end{lemma}
\begin{proof}
    Appendix \ref{noinfodomproof} contains the proof of this result.
\end{proof}

The next two lemmas elucidate some crucial properties of the equilibrium
demand function, in effective value, facing an individual firm.

\begin{lemma}
\label{no atom}In a symmetric equilibrium \(\left( \Phi ,H\right)\), the
implied effective-value distribution has no atom over \(\left[0,\bar{U}\right]\). 
\end{lemma}

\begin{proof}
Effective value \(\bar{U}\) is clearly suboptimal when price is a choice
variable as it can be generated only by a zero price coupled with full
disclosure, which yields zero profit. Lemma 4.1 of \cite{avp}
rules out any interior atom over \(\left( 0,\bar{U}\right)\) in the case of
exogenous price by showing that if there is an interior atom, firms would
have a profitable deviation to spread the weight at the atom to its
neighborhood in a mean-preserving manner. The same logic is clearly
applicable here. Note that when price is also a choice variable, the minimum
possible effective value is no longer \(0\). As effective value \(0\) is also an
interior value, the aforementioned lemma is also applicable here. \end{proof}

If other firms adopt strategy \(\left(\Phi ,\left\{ H\left( \cdot ;p\right)
\colon p\in \supp\left( \Phi \right) \right\} \right) \) with an effective value
distribution that is continuous over \(\left[ 0,\bar{U}\right]\), the
(expected) demand function \(D\left( w;\Phi ,H\right)\) in effective value
facing firm \(i\) is also continuous over \(\left[ 0,\bar{U}\right]\) and is
given by 

\[D\left( w;\Phi ,H\right) = \begin{cases}
    0, \quad &\text{if} \quad w < 0\\
    \left[ \int_{0}^{1}H\left( w;p\right) d\Phi \left( p\right) \right] ^{n-1} \quad &\text{if} \quad w \geq 0 \text{ .}
\end{cases}\]
Firm \(i\)'s profit from strategy \(\sigma_i\) is given by 
\[
\int_{0}^{1}p\left[ \int_{\ubar{U}-p}^{\bar{U}-p}D\left(w;\Phi
,H\right) dH_{i}\left( w;p\right) \right] d\Phi _{i}\left( p\right) \text{.}
\]

Our next lemma states that the equilibrium demand function must be well-behaved. Defining \(\bar{d} \equiv \sup \left\{ w\in (0,\bar{U}]:\int_{0}^{1}H\left( w;p\right) d\Phi
\left( p\right) <1\right\}\),
\begin{lemma}
\label{demand}In a symmetric equilibrium \(\left( \Phi ,H\right)\), the
demand function in effective value \(D\left( w;\Phi ,H\right)\) is continuous
over \(\left[ 0,\bar{U}\right]\) and strictly increasing over \(\left[ 0,\bar{d}\right]\). Moreover, \(D\left( 0;\Phi ,H\right) >0\).
\end{lemma}

\begin{proof}
Please see Appendix \ref{demandlemmaproof}.
\end{proof}

In addition to the continuity and strict monotonicity, the equilibrium
demand is positive at effective value \(0\). This arises when there is a
positive probability that all other firms realize posteriors below the
prices they charge.

The next lemma establishes that when the demand function has these
properties, an optimal signal takes a simple form: either it provides full
information or it provides no information.

\begin{lemma}\label{full info is opt}
\label{opt info} Consider a demand function in effective values satisfying the properties in Lemma \ref{demand}. The information policy in the optimal
strategy can only be either full information or no information. Moreover, there is always an optimal strategy involving full information. No information can emerge in the optimal strategy only if used in conjunction with price \(\ubar{U}\).
\end{lemma}

\begin{proof}
Please refer to Appendix \ref{optinfoproof}.
\end{proof}

The intuition of the lemma is as follows. Fix a price temporarily and we are
back to the setting of \cite{avp}. There, we show that fixing
the reservation value, the optimal effective-value distribution assigns a
positive mass to that reservation value (fixed persuasion budget) and the
residual mass to posteriors below according to the concavification of the
payoff function in effective value (flexible persuasion budget). As such, an
optimal effective-value distribution can be constructed with at most a
ternary support. Moreover, restricted to this ternary support of the optimal
distribution, the demand is concave. The concavity implies that the firm can
increase its expected demand if it can further contract the support to
binary values in a mean-preserving way. 

In the proof of Lemma \ref{opt info}, we detail
how this can be achieved by a reduction in the reservation value accompanied
by a price adjustment. We further exploit the concavification of the demand
to show that the optimal profit can be attained by effective value
distribution supported on \(\left\{ -p,w\right\}\), where \(p\) is the charged
price and \(w\) is the reservation value. If the demand at \(w\) is \(D\left(
w\right)\), the profit of this class of strategy is simply \(p\times \left[ 
\ubar{U}/\left( w+p\right) \right] \times D\left( w\right)\), which is
strictly increasing in \(p\). In order to maintain the reservation
value at \(w\), more information needs to be provided to make up for the price
hike. This implies the optimality of full information, in conjunction with
the highest price that is consistent with \(w\) being the reservation value,
i.e., \(\bar{U}-w\). No information used in conjunction with price \(\ubar{U}\) can arise as an optimal solution only when the demand function happens to have full information and price \(\ubar{U}\) as an optimal solution, as the two strategies yield identical profit.

An immediate consequence of Lemmas \ref{demand} and \ref{opt info}
is that \textit{the only possible information policies in a symmetric equilibrium are full information and no information}. Recall that Lemma \ref{no info
dominated} shows that no information is strictly dominated except for the
case of pricing at \(\ubar{U}\). However, we cannot have every firm
assigning a positive probability to price \(\ubar{U}\) and no information in equilibrium, as it implies an atom at \(0\) in the equilibrium effective value distribution, contradicting Lemma \ref{no atom}. As a result, in a symmetric equilibrium, full information must be chosen almost surely.

In the remainder of this section, we explicitly construct the unique
symmetric equilibrium. It remains to identify the price distribution that is
a mutual best response for firms under full information.

For a firm adopting full information and price distribution \(\Phi\) with
support contained in \(\left[ 0,\bar{U}\right]\),\footnote{It is immediate that prices above \(\bar{U}\) yield zero profit for sure and are never chosen. } its implied effective value distribution is 
\[H\left( w\right) =
\begin{cases}
    \left( 1-\mu \right) \left[ 1-\Phi \left( \bar{U}-w\right) \right] \quad  &\text{if} \quad w<0\\ 
    \left( 1-\mu \right) +\mu \left[ 1-\Phi \left( \bar{U}-w\right) \right]  \quad &\text{if} \quad w\geq 0\text{.}
\end{cases}
\]
If all but firm \(i\) adopts this strategy, the demand facing firm \(i\) is
given by
\[D\left( w\right) =
\begin{cases}
    0 \quad  &\text{if} \quad w<0\\ 
    \left[ 1-\mu \Phi \left( \bar{U}-w\right) \right] ^{n-1}  \quad &\text{if} \quad w\geq 0\text{.}
\end{cases}
\]
Therefore, by pricing at \(p\in \left[ 0,\bar{U}\right]\) and full
information, firm \(i\)'s profit is \(p\mu D\left( \bar{U}-p\right)\).

Following standard arguments, if \(\Phi\) is an equilibrium price distribution, then all prices on its support must yield the same profit to
an individual firm. Moreover, the maximum of its support must be \(\bar{U}\).
These observations allow us to pin down \(\Phi\) at%
\[\Phi \left( p\right) =\frac{1}{\mu}-\frac{1-\mu }{\mu }\left( \frac{\bar{U}}{p}%
\right) ^{\frac{1}{n-1}}, \quad \text{with support} \quad \left[\left( 1-\mu \right) ^{n-1}\bar{U},\bar{U}\right] \text{.}  \label{price distn}\tag{\(3\)}\]

It is straightforward to verify that the demand function implied by (\ref%
{price distn}) indeed satisfies all the properties in Lemma \ref{demand}. As
such, full information is indeed the optimal information policy for each
individual firm.

\begin{proposition}\label{main prop}
When both information and prices are posted, there is a unique symmetric equilibrium. In it, firms provide full information and randomize
over prices according to (\ref{price distn}).
\end{proposition}

\subsection{Alternate Sufficiency Argument for Proposition \ref{main prop}}

It is possible to prove that the unique symmetric equilibrium specified in Proposition \ref{main prop} is an equilibrium through a simple direct argument. Indeed, we make the \textit{ansatz} that there is an equilibrium in which firms provide full information, then solve for the pricing-only game, which generates the equilibrium distribution specified in (\ref{price distn}). From there, it remains to check double-deviations to alternate price-distribution combinations.

It suffices to fix an arbitrary price, then solve for optimal distribution over posterior effective values. In principle, this is quite difficult: not all Bayes-plausible distributions over effective values are feasible (as explained in \cite{avp}). However, we can solve for the optimal distribution by ignoring the feasibility constraint, solving for the (possibly infeasible) optimal effective value distribution via straightforward concavification (\cite{kam}), then checking to see whether the resulting distribution is, indeed, feasible.

There are two possible regions of prices to check. The first is when the price is in the interval \(\left[\left( 1-\mu \right) ^{n-1}\bar{U},\bar{U}\right]\). In that case, the firm's payoff as a function of the realized effective value \(w\) is
\[\Pi\left(w\right) = \begin{cases}
0, \quad &-p \leq w < 0\\
p\left(1-\mu\right)^{n-1}\left(\frac{\mu-c}{\mu-c-\mu w}\right), \quad &0 \leq w \leq \bar{U} - p\text{ .}
\end{cases}\]

It is easy to see that the line \[t(w) = \left(1-\mu\right)^{n-1}w + \left(1-\mu\right)^{n-1}p \text{ ,}\]
lies weakly above \(\Pi\) on \(\left[-p,\bar{U}-p\right]\) and intersects \(\Pi\) at the three points \(\left\{-p,0,\bar{U} - p\right\}\). The full information distribution is supported on effective values \(-p\) and \(\bar{U} - p\) and is feasible, so it is optimal. Figure \ref{figproofill} illustrates a potential deviator's payoff as a function of the realized effective value for the (boundary) price of \(p = \left( 1-\mu \right) ^{n-1}\bar{U}\) (when \(n = 2\)).

The other case is when \(p < \left( 1-\mu \right) ^{n-1}\bar{U}\). In that instance, we can proceed in a similar manner and verify that the firm's payoff is bounded above by the equilibrium payoff of \(\bar{U} \mu \left(1-\mu\right)^{n-1}\).

\begin{figure}
    \centering
    \includegraphics[scale=.1]{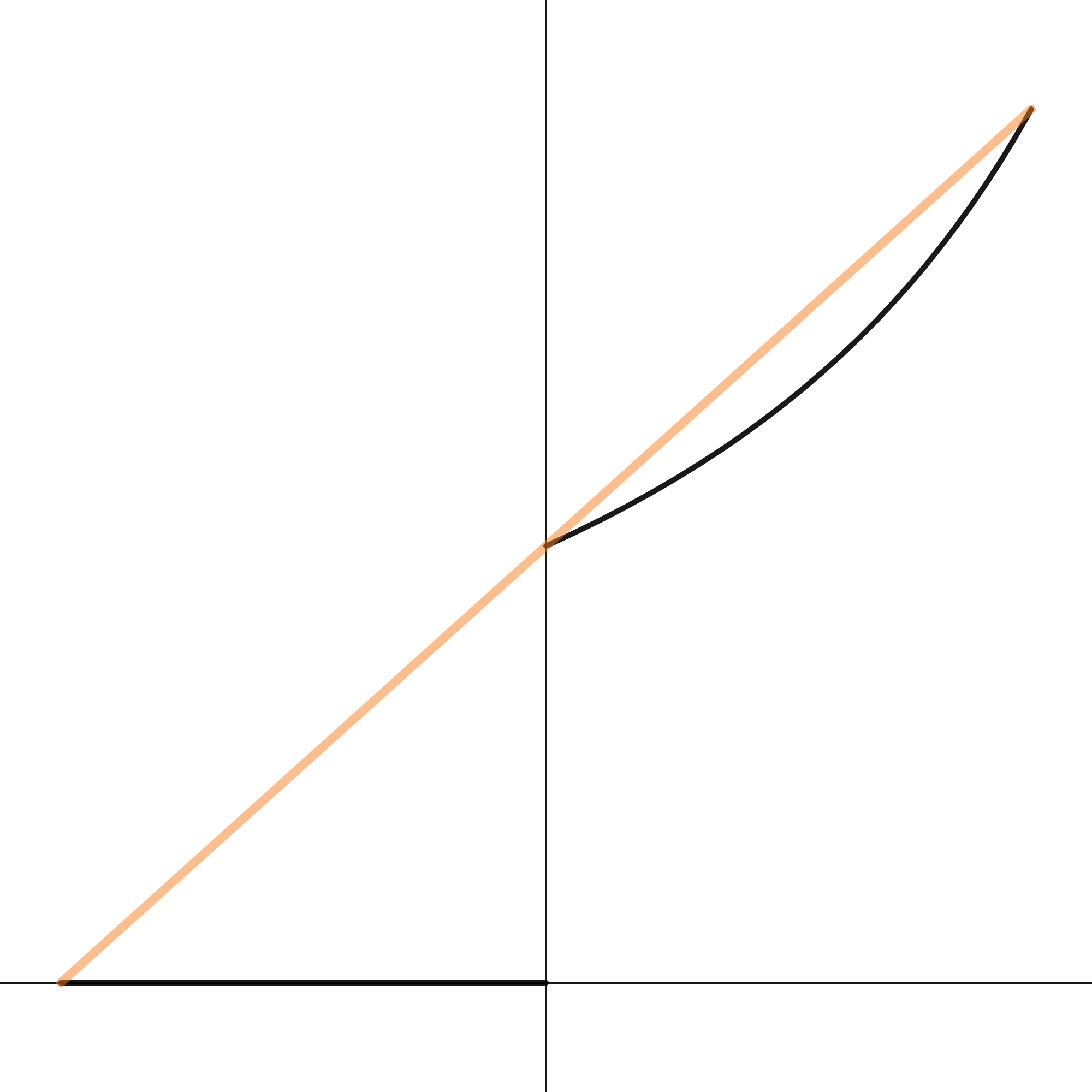}
    \caption{Payoff as a function of the realized effective value (Black) and its Concavification (Orange) for a fixed price.}
    \label{figproofill}
\end{figure}

\subsection{To Post or Not to Post}
So far in this section, we restrict attention to the setting in which firms are "exogenously required" to post their information policies. Suppose instead firms are free to choose whether to post their information policies or not. Is it in their interests to post and commit to their information policies at the outset?

In this alternative setting, a typical strategy of a firm includes not only the posted price and signal, but also whether the signal is posted or not. Focusing again on symmetric equilibrium, the following corollary states that the unique symmetric equilibrium has all firms posting their signals and hence the outcome effectively coincides with that of Proposition \ref{main prop}.

\begin{corollary}
Suppose firms can choose whether to post their signals or not. In the unique symmetric equilibrium, every firm posts full information and randomizes
over prices according to (\ref{price distn}).
\end{corollary}

By the argument in Lemma \ref{onlyonelemma}, any firm that charges a positive price and chooses to hide the information policy is believed to provide useless information. Consequently, any firm that refuses to post its signal has effectively chosen to provide useless information. By Lemma \ref{full info is opt}, this is optimal for a firm only if price \(\ubar{U}\) is charged. However, the symmetric equilibrium cannot have firms assigning a strictly positive probability at price \(\ubar{U}\) as it implies an atom in the equilibrium effective value distribution. Consequently, by Lemma \ref{full info is opt} again, the symmetric equilibrium must have every firm posting full information almost surely.

\section{Comparative Statics}

\(\S\)\ref{hiddeninfosection} reveals that when information is hidden, each firm obtains zero profit. The consumer's welfare is simply \(\mu-c\). As noted in \(\S\)\ref{postedinfosection}, when information is posted, each firm obtains a
profit of 
\[
\bar{U}\times \mu \times \underbrace{\left( 1-\mu \right) ^{n-1}}_{\text{Demand at effective value } 0} =\left( \mu -c\right) \left( 1-\mu
\right) ^{n-1}\text{ .}
\]%
The consumer's welfare is given by the expected maximum effective value. By
straightforward computation, it is given by%
\[
u^{PI}\left( n\right) \equiv \frac{\mu -c}{\mu }\left( 1-n\mu \left( 1-\mu
\right) ^{n-1}-\left( 1-\mu \right) ^{n}\right) \text{ .}
\]

It is clear that the firms always strictly prefer the posted information setting (a positive number is bigger than \(0\)). Information posting allows firms to overcome the "informational Diamond paradox," thus activating product differentiation in the market. With product differentiation, a firm can secure a positive expected profit, as there is always a chance that its product is the only one that can match the consumer's preference.

The consumer's welfare is, however, less clear-cut. In the posted
information setting, the consumer has access to full information on all
products, allowing her to find a product that matches well with her
preference. However, the product differentiation produces milder price
competition, so she has to pay a higher price. Whether the benefit of better
information can outweigh the cost of higher  prices depends on the number of firms in the market. When there are few firms in the market, the
benefit of full information is limited (as she is not likely to find a
well-matched product), whereas the cost of higher prices is significant (as
price competition is mild). In this case, the consumer's welfare is hurt by
information posting. However, when there are many firms in the market, the
benefit of full information is huge (as she is likely to find a well-matched
product), whereas the cost of higher price is limited (as price competition
is intense). In this case, the consumer can strictly gain from information
posting.

\begin{proposition}
\label{comp stat} Consumers benefit from information posting if and only if the number of firms is sufficiently large.
\end{proposition}

\begin{proof}
See Appendix \ref{compstatproof}.
\end{proof}

\section{Discussion}

We show in this paper that consumers may be better off if firms can only guide search through their prices and not through the information they promise to provide. With hidden information, firms provide no information and there is no search. Nevertheless, the corresponding lack of product differentiation engenders fierce price competition. When the market is thin, this is to consumers' benefits, even though the probability of consumers obtaining a high match value is the lowest possible.

Going forward, a natural question is whether these results can be generalized to other match-value distributions, in particular the case in which the distribution possesses a density. With hidden information, it is straightforward to see that our result persists: firms provide no useful information and price at marginal cost (\(0\)).\footnote{None of the arguments for the results of \(\S\)\ref{hiddeninfosection} rely on the binary distribution over match values. The subsumed paper, \cite{dia}, contains more details and discussion of this.} Things are much more difficult with posted information. In particular, even characterizing the equilibrium distribution over prices for a particular conjectured match value distribution is a challenge, since firms necessarily mix.\footnote{By imposing consumer heterogeneity, \cite*{choi2} negate this issue with aplomb.} From there, one must then check whether there are no double deviations by firms.

\bibliography{sample.bib}

\appendix

\section{Omitted Proofs}

\subsection{Lemma \ref{onlyonelemma} Proof}

\begin{proof}
    It is obvious that there can be no equilibria in which the consumer does not search with a probability greater than \(0\). A distribution and price in support of a firm's mixed strategy that is followed by no search with strictly positive probability necessarily induce a reservation value that is strictly less than \(0\). Consequently, a firm could always deviate profitably by instead charging a very low but positive price and providing no information.

    It suffices to show that in any symmetric equilibrium, firms provide useless information with probability \(1\), as by the definition of useless information, the consumer necessarily stops her search and buys from any visited firm.
    
    Suppose for the sake of contradiction that there exists a symmetric equilibrium in which useful information is provided with a strictly positive probability. By the definition of useful information, the consumer searches more than one firm with a positive probability.

    There are two possibilities, either firms make strictly positive profits or they make zero profits.

    \medskip
    
    \noindent \textbf{Case 1. Strictly positive profits.} By assumption, for each firm \(i\), there exists a price, conjectured distribution pair \(\left(p_i, \tilde{F}_i\right)\) in support of its mixed strategy, with \(p_i > 0\), such that if firm \(i\) chooses \(\left(p_i, \tilde{F}_i\right)\), 
\begin{enumerate}[noitemsep,topsep=0pt]
    \item Firm \(i\) is visited with strictly positive probability;
    \item Firm \(i\)'s induced reservation value satisfies \(\tilde{U}_i > \ubar{U} - p_i\) (as the information contained in \(tilde{F}_i\) is useful) and \(\tilde{U}_i \geq 0\) (as firm \(i\) is visited); and
    \item \(\tilde{F}_i\) places strictly positive probability on posteriors \(x\) at which the consumer purchases from firm \(i\) with probability that is strictly less than \(1\), conditional on visiting firm \(i\) (again as the information is useful). Denote this set of posteriors \(\tilde{X}\), and note that \(\tilde{X} \subseteq \left[0,\tilde{U}_i + p_i\right)\).
\end{enumerate}
By the reservation value formula, \(\tilde{F}_i\left(\tilde{U} + p_i\right) < 1\), and by assumption \(\tilde{F}_i\left(\sup \tilde{X}\right) > 0\). Immediately we see that firm \(i\) can strictly improve her payoff by pooling the beliefs lying strictly above \(\tilde{U}_i + p_i\) with a subset of \(\tilde{X}\) on which \(\tilde{F}\) places strictly positive probability. In principle, we could try to ``punish'' firm \(i\) by altering tie-breaks at beliefs \(x < \tilde{U}_i + p_i\) in the other firms' favor. However, this is inconsequential: firm \(i\) can preserve any favorable tie-breaks at beliefs \(x < \tilde{U}_i + p_i\) by increasing any such belief by \(\varepsilon > 0\) (small) by averaging them with a collection of beliefs strictly above \(\tilde{U}_i + p_i\) on which \(\tilde{F}_i\) places vanishingly small probability.

Consequently, a symmetric equilibrium necessarily has \(\tilde{F}_i\left(\sup \tilde{X}\right) = 0\), and the posterior realizations have no effect on the consumer's purchase decision. The information contained in the signal is thus necessarily useless. 

 \medskip
    
    \noindent \textbf{Case 2. Zero profits.} Consider next symmetric equilibrium in which firms earn zero profit. This requires all firms pricing at \(0)\) with probability \(1\). In order that any deviation to a positive price is unprofitable, it is necessary that for all \(p_i > 0\), firm \(i\) is visited with probability \(0\), for otherwise it could secure a positive expected demand and hence a positive profit (say by full information). As the consumer's belief about the reservation value of firm \(i\) is at least \(\ubar{U} - p_i\), firm \(i\) is never visited only if the effective value realization of firms other than firm \(i\) is no less than \(\ubar{U} - p_i\) with probability \(1\). As \(p_i\) can be vanishingly small, this implies the effective value realization of firms other than firm \(i\) is no less than \(\ubar{U}\) with probability \(1\). It is thus necessary that these firms charge a zero price and provide useless information. By "no signaling what you don't know," the same is true if firm \(i\) follows the equilibrium strategy to price at \(0\). \end{proof}
    



\subsection{Lemma \ref{no info dominated} Proof}\label{noinfodomproof}
\begin{proof}
The effective-value distributions induced by other firms' strategies \(\sigma_{-i}\) define a demand function, denoted by \(D_{i}\left( w;\sigma
_{-i}\right) \), for firm \(i\) in the effective value \(w\) she generates.
This demand function specifies the expected probability of winning the
effective value competition when firm \(i\) realizes \(w\) and other firms play strategies \(\sigma_{-i}\). It is immediate that \(D\left( w;\sigma_{-i}\right) =0\) for all \(w<0\) and for all \(\sigma_{-i}\), as the consumer has an outside option of zero.

Take any arbitrary strategy \(\sigma_{-i}\) of the other firms. Firm \(i\) prices
at \(p\) and provides no information (i.e., \(H_{i}\left( \cdot ;p\right)\) is
degenerate at \(\ubar{U}-p\)), and its expected profit is, therefore, \(p\times D\left( 
\ubar{U}-p;\sigma _{-i}\right) \). This profit is zero if \(p>\ubar{U}\), so the strategy is clearly strictly dominated (by say full information
and price at \(\frac{ \bar{U}-\ubar{U}}{2}\)). Let's focus on \(p\leq \ubar{U}\).

Suppose, instead, firm \(i\) provides full information and prices at \(p^{\prime
}\equiv \bar{U}-\ubar{U}+p\). This yields the expected profit
\[\begin{split}
   p^{\prime }\times \mu \times D\left( \bar{U}-p^{\prime };\sigma _{-i}\right) &= \left( \bar{U}-\ubar{U}+p\right) \mu \times D\left( \ubar{U}-p;\sigma_{-i}\right)\\
   &= \left[p+\left( 1-\mu \right) \left( \ubar{U}-p\right) \right]\times D\left( \ubar{U}-p;\sigma_{-i}\right)\\
   &\geq p D\left(\ubar{U}-p;\sigma _{-i}\right) \text{,}
\end{split}\]
where the inequality is strict whenever \(p=\ubar{U}\). 

Consequently, the strategy with no information and a price distribution \(\Phi
_{i}\in \Delta \left[ 0,1\right]\) that does not assign probability one
to \(\ubar{U}\) is strictly dominated by the strategy with price
distribution \(\Phi _{i}\left( p-\left( \bar{U}-\ubar{U}\right) \right)\)
and full information. The domination is weak if \(\Phi_{i}\) is degenerate at 
\(\ubar{U}\). We do not obtain strict domination because there exists 
\(\sigma _{-i}\) such that no information and pricing at \(\ubar{U}\) is a
best response for firm \(i\).\end{proof}

\subsection{Lemma \ref{demand} Proof}\label{demandlemmaproof}
\begin{proof}
The continuity of the equilibrium demand function follows
from the continuity of the equilibrium effective-value distribution chosen
by the firms.

Suppose \(\ubar{d} \equiv \min \left\{ w:D\left( w;\Phi ,H\right)
>0\right\}\) is strictly positive. As \(D\left( w;\Phi ,H\right)\) is
continuous at \(\ubar{d}\), it is locally convex at \(\ubar{d}\),
i.e., there is a neighborhood of \(\ubar{d}\) such that \(D\left( w;\Phi
,H\right)\) lies strictly below its concavification. Hence it is suboptimal to assign any positive weight on this neighborhood, contradicting the
definition of \(\ubar{d}\). Therefore, \(\ubar{d}=0\). 

Next, recall that the equilibrium profit and, hence, prices are positive. For any \(p\in \supp\left(\Phi \right)\), the possible range of effective values, \(\left[ -p,\bar{U}-p\right]\), contains \(0\). As \(D\left( w;\Phi ,H\right) =0\) for all \(w\leq 0\),
if \(D\left( 0;\Phi ,H\right) =0\), then the demand is locally convex at \(0\),
and assigning weight in some neighborhood of \(0\) is suboptimal, which again
is a contradiction.

It is immediate that \(D\left( w;\Phi ,H\right)\) is weakly increasing over \(\left[ 0,\bar{d}\right]\). If the demand contains any flat segment, then it
contains some maximal flat segment, say nonempty interval \(\left[ a,b\right] \subset \left[ 0,\bar{d}\right]\), such that \(D\left( a;\Phi ,H\right)
=D\left( b;\Phi ,H\right)\) and \[D\left( a-\varepsilon ;\Phi ,H\right)
<D\left( a;\Phi ,H\right) =D\left( b;\Phi ,H\right) <D\left( b+\varepsilon
;\Phi ,H\right)\text{,}\] for any \(\varepsilon >0\). However, the demand function is locally convex at \(b\), contradicting the maximality of the flat segment.\end{proof}

\subsection{Lemma \ref{opt info} Proof}\label{optinfoproof}

\begin{proof}
Let \(D\left( w\right)\) be the demand function facing an individual firm in its realized effective value \(w\). The existence of the best response follows
from the continuity of the continuity of the objective and the compactness
of the strategy space. Denote by \(\Pi^{\ast}\) the optimal profit.

\begin{claim}
The optimal profit \(\Pi^{\ast}\) can be attained by an effective-value
distribution with a support of no more than 2 values.
\end{claim}

\begin{proof}
As shown in \(\S\)3.3.1 in \cite{avp}, when
the price is fixed, there is an optimal effective-value distribution with support on no more than \(3\) values. It is therefore without loss to suppose an optimal strategy when firms can choose price consists of a price \(p^{\ast}\) and an effective-value distribution \(H^{\ast}\) supported on \(\left\{
w_{0},w_{1},w_{2}\right\} \subset \left[ -p^{\ast },\bar{U}-p^{\ast }\right]\). If \(\left\{ w_{0},w_{1},w_{2}\right\}\) are not distinct, we are done, so suppose \(w_{0}<w_{1}<w_{2}\). Note that \(w_{2}\), the maximum effective
value on the support, corresponds to the reservation value of the firm's
signal. Moreover, the weight assigned to the respective effective values are 
\[H^{\ast }\left( w_{0}\right)  = \left( 1-\frac{c}{1-\left( w_{2}+p^{\ast
}\right) }\right) \times \frac{w_{1}-a_{p^{\ast }}\left( w_{2}\right) }{w_{1}-w_{0}}\text{,}\]
\[H^{\ast }\left( w_{1}\right) -H^{\ast }\left( w_{0}\right)  = \left( 1-\frac{c}{1-\left( w_{2}+p^{\ast }\right) }\right) \times \frac{a_{p^{\ast
}}\left( w_{2}\right) -w_{0}}{w_{1}-w_{0}}\text{,}\]
and
\[H^{\ast }\left( w_{2}\right) -H^{\ast }\left( w_{1}\right)  = \frac{c}{1-\left( w_{2}+p^{\ast }\right) }\text{,}\]
where
\[\begin{split}
    a_{p^{\ast }} \colon \left[ \ubar{U}-p^{\ast },\bar{U}%
-p^{\ast }\right] &\rightarrow \left[ -p^{\ast },\ubar{U}-p^{\ast }\right]\\
w_2 &\underset{a_{p^{\ast}}}{\mapsto} \frac{\left( \mu -c-p^{\ast }\right) %
\left[ 1-\left( w_{2}+p^{\ast }\right) \right] -cw_{2}}{1-\left(
w_{2}+p^{\ast }\right) -c}
\end{split}\]
Consequently, the optimal profit \(\Pi^{\ast}\) is given by%
\[
\Pi ^{\ast }\equiv p^{\ast }\times \left[ \left( 1-\frac{c}{1-\left(
w_{2}+p^{\ast }\right) }\right) \left[ \frac{w_{1}-a_{p^{\ast }}\left(
w_{2}\right) }{w_{1}-w_{0}}D\left( w_{0}\right) +\frac{a_{p^{\ast }}\left(
w_{2}\right) -w_{0}}{w_{1}-w_{0}}D\left( w_{1}\right) \right] +\frac{c}{%
1-\left( w_{2}+p^{\ast }\right) }D\left( w_{2}\right) \right] \text{.}
\]%
Furthermore, it is without loss to suppose that on the graph of the demand function, \(\left\{ \left( w_{0},D\left( w_{0}\right) \right) ,\left(
w_{1},D\left( w_{1}\right) \right) ,\left( w_{2},D\left( w_{2}\right)
\right) \right\} \) do not lie on a straight line. If this were not the case,
an effective-value distribution supported on \(\left\{ w_{0},w_{2}\right\}\)
coupled with price \(p^{\ast}\) would already generate the optimal profit \(\Pi^{\ast }\), so we are done.

We consider the two cases in turn: 1. \(w_0 \geq 0\) and 2. \ \(w_0 < 0\).

\medskip

\noindent

\textbf{Case 1. \(w_0 \geq 0\).} This implies \(D\left( w_{0}\right) >0\), and as
\(\left\{ \left( w_{0},D\left( w_{0}\right) \right) ,\left(
w_{1},D\left( w_{1}\right) \right) ,\left( w_{2},D\left( w_{2}\right)
\right) \right\}\) do not lie on a straight line,
\[
\frac{D\left( w_{2}\right) }{w_{2}+p^{\ast }}<\frac{D\left( w_{1}\right) }{%
w_{1}+p^{\ast }}\leq \frac{D\left( w_{0}\right) }{w_{0}+p^{\ast }}\text{.}
\]%
Using this inequality, the optimal profit \(\Pi^{\ast}\) must be strictly less than
\[\begin{split}
    &p^{\ast }\times \left( 1-\frac{c}{1-\left(
w_{2}+p^{\ast }\right) }\right) \left( \frac{w_{1}-a_{p^{\ast }}\left(
w_{2}\right) }{w_{1}-w_{0}}D\left( w_{0}\right) +\frac{a_{p^{\ast }}\left(
w_{2}\right) -w_{0}}{w_{1}-w_{0}}\frac{w_{1}+p^{\ast }}{w_{0}+p^{\ast }}%
D\left( w_{0}\right) \right)\\ &+ p^{\ast }\times\left(\frac{c}{1-\left( w_{2}+p^{\ast }\right) }%
\frac{w_{1}+p^{\ast }}{w_{0}+p^{\ast }}D\left( w_{0}\right)\right)\\
&= p^{\ast }\times \frac{\mu -c}{p^{\ast }+w_{0}}D\left( w_{0}\right)\text{ .}
\end{split}\]

If \(w_{0}\geq \ubar{U}-p^{\ast }\), this is the profit the firm can secure by offering price \(p^{\ast}\) coupled with a binary
effective-value distribution supported on \(\left\{ -p^{\ast },w_{0}\right\}\). If \(w_{0}<\ubar{U}-p^{\ast }\), then 
\[
\Pi ^{\ast }<p^{\ast }\times \frac{\mu -c}{p^{\ast }+w_{0}}D\left(
w_{0}\right) =\left( \left( \mu -c-w_{0}\right) -w_{0}\frac{\mu -c-p^{\ast
}-w_{0}}{p^{\ast }+w_{0}}\right) D\left( w_{0}\right) <\left( \mu
-c-w_{0}\right) D\left( w_{0}\right) \text{,}
\]%
where the last expression is the profit that the firm can secure by offering
price \(\mu -c-w_{0}\) coupled with an effective-value distribution degenerate
at \(w_{0}\) (i.e., no information). Both scenarios contradict the optimality
of \(\Pi^{\ast}\).

\medskip

\noindent \textbf{Case 2. \(w_0 < 0\).} As \(D\left( w\right) =0\) for all \(w<0\) and \(D\left( 0\right)
>0\), \(H^{\ast}\) is optimal only if \(w_{0}=-p^{\ast }\). Moreover, \(D\left(
w_{0}\right) =0\), and the non-collinearity of \(\left\{ \left( w_{0},D\left(
w_{0}\right) \right) ,\left( w_{1},D\left( w_{1}\right) \right) ,\left(
w_{2},D\left( w_{2}\right) \right) \right\}\) implies
\[
\frac{D\left( w_{2}\right) }{w_{2}+p^{\ast }}<\frac{D\left( w_{1}\right) }{%
w_{1}+p^{\ast }}\text{.}
\]%
Using this inequality and following steps similar to the above, 
\[
\Pi ^{\ast }=p^{\ast }\times \left[ \left( 1-\frac{c}{1-\left( w_{2}+p^{\ast
}\right) }\right) \frac{a_{p^{\ast }}\left( w_{2}\right) +p^{\ast }}{%
w_{1}+p^{\ast }}D\left( w_{1}\right) +\frac{c}{1-\left( w_{2}+p^{\ast
}\right) }D\left( w_{2}\right) \right] <p^{\ast }\times \frac{\mu -c}{%
p^{\ast }+w_{1}}D\left( w_{1}\right) \text{.}
\]

If \(w_{1}\geq \ubar{U}-p^{\ast }\), the last expression is the profit
that the firm can secure by offering price \(p^{\ast}\) coupled with a binary
effective-value distribution supported on \(\left\{ -p^{\ast },w_{1}\right\}\). If \(w_{1} < \ubar{U}-p^{\ast }\), 
\[
\Pi ^{\ast }<p^{\ast }\times \frac{\mu -c}{p^{\ast }+w_{1}}D\left(
w_{1}\right) =\left( \left( \mu -c-w_{1}\right) -w_{1}\frac{\mu -c-p^{\ast
}-w_{1}}{p^{\ast }+w_{1}}\right) D\left( w_{1}\right) <\left( \mu
-c-w_{1}\right) D\left( w_{1}\right) \text{,}
\]%
which is the profit that the firm can secure by offering price \(\mu -c-w_{1}\)
coupled with an effective-value distribution degenerate at \(w_{1}\) (i.e., no
information). Both scenarios contradict the optimality of \(\Pi^{\ast}\).
\end{proof}

Suppose the optimal profit \(\Pi^{\ast}\) can be attained by price \(p^{\ast}\) and an effective-value distribution \(H^{\ast}\) supported on \(\left\{
w_{0},w_{2}\right\}\). It remains to show that if \(H^{\ast}\) is not
degenerate, i.e., if \(w_{0}<w_{2}\), then it must have \(w_{0}=-p^{\ast }\) and \(w_{2}=\bar{U}-p^{\ast }\).

Suppose \(H^{\ast}\) is not degenerate. Then \(w_{0}<\mu -c-p^{\ast }<w_{2}\),
as the mean of \(H^{\ast}\) is \(\mu -c-p^{\ast }\). We first show that an
effective-value distribution supported on \(\left\{ -p^{\ast },w_{2}\right\}\)
together with price \(p^{\ast}\) can also attain \(\Pi^{\ast}\).

Suppose the effective-value distribution supported on \(\left\{ -p^{\ast
},w_{2}\right\}\) does not yield the optimal profit. Then it is necessary
that \(w_{0}>0\) and \(D\left( w_{0}\right) >0\). Moreover, 
\[
\frac{D\left( w_{2}\right) }{w_{2}+p^{\ast }}<\frac{D\left( w_{0}\right) }{%
w_{0}+p^{\ast }}\text{,}
\]%
for otherwise the effective value distribution supported on \(\left\{
-p^{\ast },w_{2}\right\}\) already yields a strictly higher profit than that
supported on \(\left\{ w_{0},w_{2}\right\}\). Using this inequality, the
optimal profit can be bounded as follows.
\[\begin{split}
    \Pi ^{\ast } &= p^{\ast }\times \left[ \frac{w_{2}-\left( \mu -c-p^{\ast}\right) }{w_{2}-w_{0}}D\left( w_{0}\right) +\frac{\left( \mu -c-p^{\ast}\right) -w_{0}}{w_{2}-w_{0}}D\left( w_{2}\right) \right]\\
    &< p^{\ast }\times \left[ \frac{w_{2}-\left( \mu -c-p^{\ast }\right) }{w_{2}-w_{0}}+\frac{\left( \mu -c-p^{\ast }\right) -w_{0}}{w_{2}-w_{0}}\frac{w_{2}+p^{\ast }}{w_{0}+p^{\ast }}\right] D\left( w_{0}\right)\\
    &= p^{\ast }\times \frac{\mu -c}{p^{\ast }+w_{0}}D\left( w_{0}\right)\\
    &= \left( \left( \mu -c-w_{0}\right) -w_{0}\frac{\mu -c-p^{\ast }-w_{0}}{p^{\ast }+w_{0}}\right) D\left( w_{0}\right) < \left( \mu -c-w_{0}\right) D\left( w_{0}\right)\text{,}
\end{split}\]
which is the profit that the firm can secure by offering price \(\mu -c-w_{0}\)
coupled with an effective-value distribution degenerate at \(w_0\) (i.e., no
information). This contradicts the optimality of \(\Pi^{\ast}\).

Finally, we show that \(w_{2}=\bar{U}-p^{\ast }\), i.e., \(H^{\ast}\) provides
full information. Suppose to the contrary that \(w_{2}<\bar{U}-p^{\ast }\). The optimal profit is 
\[
\Pi ^{\ast }=p^{\ast }\times \frac{\mu -c}{w_{2}+p^{\ast }}D\left(
w_{2}\right) \text{.}
\]%
As \(w_{2}<\bar{U}-p^{\ast }\), the firm can raise the price to some \(p^{\prime }>p^{\ast }\) while maintaining the support of the effective value
distribution at \(\left\{ -p^{\prime },w_{2}\right\}\). The resulting profit
is%
\[
p^{\prime }\times \frac{\mu -c}{w_{2}+p^{\prime }}D\left( w_{2}\right)
>p^{\ast }\times \frac{\mu -c}{w_{2}+p^{\ast }}D\left( w_{2}\right) =\Pi
^{\ast }\text{,}
\]%
contradicting the optimality of \(\Pi^{\ast}\).
\end{proof}

\subsection{Proposition \ref{comp stat} Proof}\label{compstatproof}
\begin{proof}
The consumer's welfare under perfect information \(u^{PI}\left( n\right)\) is
increasing in \(n\), as%
\begin{eqnarray*}
\frac{du^{PI}\left( n\right) }{dn} &=&-\frac{\mu -c}{\mu }\left( 1-\mu
\right) ^{n-1}\left[ \mu +\left( 1-\mu +n\mu \right) \ln \left( 1-\mu
\right) \right]  \\
&\geq &-\frac{\mu -c}{\mu }\left( 1-\mu \right) ^{n-1}\left[ \mu +\left(
1-\mu +2\mu \right) \ln \left( 1-\mu \right) \right]  \\
&=&-\frac{\mu -c}{\mu }\left( 1-\mu \right) ^{n-1}\left[ \mu +\left( 1+\mu
\right) \ln \left( 1-\mu \right) \right]  \\
&>&0\text{,}
\end{eqnarray*}%
where the last inequality follows from the fact that the term \(\left[ \mu
+\left( 1+\mu \right) \ln \left( 1-\mu \right) \right]\) is decreasing in \(\mu\) and equals \(0\) at \(\mu =0\). The proposition follows by noting that \(u^{PI}\left( 2\right) =\mu \left( \mu -c\right) <\mu -c\) and \(\lim_{n\rightarrow \infty }u^{PI}\left( n\right) = \frac{\mu-c}{\mu}
>\mu -c\).\end{proof}

\end{document}